\documentclass[letterpaper, 10pt, conference]{ieeeconf} 

\IEEEoverridecommandlockouts 

\usepackage{cite}
\usepackage{amsmath,amssymb,amsfonts}
\usepackage{graphics}
\usepackage{epsfig} 

\usepackage[utf8]{inputenc}
\usepackage{hyperref} 
\usepackage{algorithm}
\usepackage{algpseudocode}

\algrenewcommand\algorithmicrequire{\textbf{Input:}}
\algrenewcommand\algorithmicensure{\textbf{Output:}}

\newtheorem{theorem}{Theorem}[section]

\newtheorem{proposition}[theorem]{Proposition}

\title{\LARGE \bf Sparse Dynamic Network Reconstruction Through $L_1$-regularization of a Lyapunov Equation}

\author{
    Ian Xul Belaustegui$^{a}$,
    Marcela Ordorica Arango$^{a}$,
    Roman Rossi-Pool$^b$, \\
    Naomi Ehrich Leonard$^{a}$,
    Alessio Franci$^{c}$
    \thanks{$^{a}$ Department of Mechanical and Aerospace Engineering, Princeton University; \texttt{\{ianxul,m.ordorica,naomi\}@princeton.edu}}
    \thanks{$^{b}$ Instituto de Fisiología Celular - Neurociencias and Centro de Ciencias de la Complejidad, Universidad Nacional Autónoma de México, México; \texttt{romanr@ifc.unam.mx}}
    \thanks{$^{c}$ Department of Electrical Engineering and Computer Science, University of Liege, and WEL Research Institute, Wavre, Belgium; \texttt{afranci@uliege.be}}
    \thanks{This research was supported in part by AFOSR grant FA9550-24-1-0002.}
}

\begin{document}

\maketitle
\thispagestyle{empty}
\pagestyle{empty} 


\begin{abstract}
    An important problem in many areas of science is that of recovering interaction networks from high-dimensional time-series of many interacting dynamical processes. A common approach is to use the elements of the correlation matrix or its inverse as proxies of the interaction strengths, but the reconstructed networks are necessarily undirected. Transfer entropy methods have been proposed to reconstruct directed networks, but the reconstructed network lacks information about interaction strengths. We propose a network reconstruction method that inherits the best of the two approaches by reconstructing a directed weighted network from noisy data under the assumption that the network is sparse and the dynamics are governed by a linear (or weakly-nonlinear) stochastic dynamical system. The two steps of our method are {\it i)} constructing an (infinite) family of candidate networks by solving the covariance matrix Lyapunov equation for the state matrix and {\it ii)} using $L_1$-regularization to select a sparse solution. We further show how to use prior information on the (non)existence of a few directed edges to dramatically improve the quality of the reconstruction.
\end{abstract}


\section{Introduction}
\label{sec:introduction}

In many areas of science it is common to record the evolution in time of many interacting processes. A natural and important problem is then to develop data-driven methods that can recover, or reconstruct, from the resulting high-dimensional time-series the underlying interaction network. Ideally, the reconstructed network would include information about both the direction and the strength of the interactions.
Interaction networks reconstruction has applications, e.g., in neuroscience (between neurons or neuronal populations),
in ecology (i.e., between species populations in a community), in molecular biology (i.e., between  expression of genes), and in social sciences (i.e., between opinionated individuals).

An established approach, commonly employed in neuroscience applications~\cite{Liégeois_Santos_Matta_Ville_Sayed_2020,Lin_Das_Krishnan_Bazhenov_Sejnowski_2017,Mohanty_Sethares_Nair_Prabhakaran_2020}, is to use correlation measures to approximate interaction networks through the entries of the covariance matrix or its inverse, the precision matrix. Networks reconstructed from the precision matrix were found to better correlate to the underlying anatomical network~\cite{Liégeois_Santos_Matta_Ville_Sayed_2020}. However, whether using the correlation or the precision matrix, the reconstructed network will necessarily be undirected because both matrices are symmetric.

In order to reconstruct directed network interactions, transfer entropy (TE)~\cite{Schreiber_2000} has emerged as a powerful information-theoretical measure of directed information transfer between dynamical processes, with many applications in neuroscience \cite{Lobier_Siebenhühner_Palva_Palva_2014,Ursino_Ricci_Magosso_2020,Novelli_Wollstadt_Mediano_Wibral_Lizier_2019, Timme_Lapish_2018,Sharma_Bucci_Brahma_Varshney_2017}. However, transfer entropy methods can only reconstruct unweighted directed interaction networks, i.e., they can only provide information about whether a process directly affects (or not) another but not the strength or sign of this directed interaction.

A commonly used, mathematically tractable approximation is to assume that the network dynamics can be modelled by a linear (or weakly nonlinear) stochastic differential equation driven by a white noise vector. When used to describe resting-state brain activity, the resulting model is known as the (linear) Dynamical Causal Model (DCM)~\cite{friston2003dynamic,casti2023dynamic}. Under this assumption, the covariance matrix of the model is the solution to a Lyapunov equation that links it to the state matrix of the network~\cite{socha2007linearization, appliedSDEs}. If the state matrix is assumed to be symmetric (i.e., the network is undirected) and the noise components are uncorrelated, then solving the Lyapunov equation for the state matrix returns exactly the precision matrix~\cite{Liégeois_Santos_Matta_Ville_Sayed_2020}. But if no symmetry assumptions are made, then solving the Lyapunov equation for the state matrix returns a whole affine subspace of candidate solutions. Methods using differential covariance~\cite{Lin_Das_Krishnan_Bazhenov_Sejnowski_2017} enable singling out asymmetric solutions, but these methods require computing the time-derivative of the measured time-series, which can problematic or unfeasible when measurement or intrinsic noise are high. 

We introduce a novel method to reconstruct weighted, thus directed, dynamical networks, using the idea that under the extra assumption that the interaction network is sparse, $L_1$-regularization can effectively select such a sparse solution from among the candidates. We also show how prior knowledge of edge existence can be included in the constraints of the resulting linear programming (LP) problem to enhance the method performance.

The first and main contribution of the paper is to formulate an $L_1$-regularization problem for sparse network reconstruction as an LP optimization problem over the affine solution space of all state matrices that solve the covariance Lyapunov equation for a given covariance matrix. Crucially, this requires the preliminary construction of an isomorphism from such a solution space to the solution space of an underdetermined linear system of equations. The second contribution is to show how to introduce prior knowledge on edge existence into the constraints of the LP problem. The third contribution is to suggest the use of a TE-based network inference algorithm as a way of obtaining some priors on edge existence and show that it significantly improves the quality of network reconstruction. 


The paper is structured as follows. Notation is introduced in Section~\ref{sec:notation}. The network model and its interpretation are introduced in Section~\ref{sec:model}. Section~\ref{sec:LyapEqSol} studies the geometry of the solution space of the covariance Lyapunov equation for the state matrix and introduces the isomorphism needed for the LP formulation of an $L_1$-regularization problem over this space. Section~\ref{sec:L1Reg} explicitly describes the LP formulation of the $L_1$-regularization problem. Section \ref{sec:AddPriors} shows how prior information about edge existence can be incorporated into the constraints of the LP problem. Section \ref{sec:TE} introduces an algorithm based on TE to provide prior knowledge on edge existence and show how to use it to constrain the LP problem. Finally, Section \ref{sec:Val} presents extensive numerical validation results obtained over large families of randomly generated sparse Hurwitz matrices. Conclusions and perspectives are presented in Section~\ref{sec: conclusions}.

\section{Notation}
\label{sec:notation}
We denote $[n]=\{1,...,n\}$. Given $n,m\in\mathbb{N}$, we write $\mathbb{R}^{m\times n}$ for the set of $m$ by $n$ matrices with real entries. For $i\in [m]$, and $j\in [n]$, the value at the $i$-th row and $j$-th column of $A\in\mathbb{R}^{m\times n}$ is $A_{i,j}$. The column-stack vectorization isomorphism $\text{vec}:\mathbb{R}^{m\times n}\to \mathbb{R}^{mn}$ is given by $(\text{vec}(A))_{m(j-1)+i}=A_{i,j}$. Vector $\mathbf{1}_{n}\in\mathbb{R}^n$ consists of $n$ entries equal to $1$ (and similarly for $\mathbf{0}_{n}$). The set of eigenvalues of $B\in\mathbb{R}^{n\times n}$ is $\sigma(B)$. The real part of $z\in\mathbb{C}$ is denoted $\Re(z)$. For a pair of matrices $A\in\mathbb{R}^{n\times n}, B\in\mathbb{R}^{n\times n}$, $\otimes$ is the Kronecker product given by $\forall r,v,s,w\in[n], (A\otimes B)_{n(r-1)+v, n(s-1)+w} = (A_{r,s})(B_{v,w})$, so that $A\otimes B\in\mathbb{R}^{n^2\times n^2}$.

\section{Model formulation}
\label{sec:model}

We consider a network model described by the linear stochastic differential equation (SDE)~\cite{appliedSDEs}
\begin{equation}
    \label{eq:lin_sys}
    \frac{dx}{dt}(t) = Ax(t) + w(t)
\end{equation}
where $w$ is a white noise vector (with noise covariance $Q=I$) and $A\in\mathbb{R}^{n\times n}$ is a Hurwitz matrix, i.e., all of its eigenvalues have a strictly negative real part. Any solution $x$ for  system \eqref{eq:lin_sys} is a stochastic process with mean $0$ and covariance matrix $\Gamma = \mathbf{E}(xx^T)$ (see for example Section $4.3$ in\cite{appliedSDEs}).

In a network modeling setting, $A$ is interpreted as a weighted and signed adjacency matrix. A non-zero element $A_{ij}\neq0$ is interpreted as the existence of a directed edge from node $j$ to node $i$. The sign of $A_{ij}$ determines the nature of the interaction (excitatory or inhibitory) and $|A_{ij}|$ determines its strength.

\section{Solving the covariance Lyapunov equation for the state matrix}
\label{sec:LyapEqSol}

\begin{figure}
    \centering
    \includegraphics[width=0.49\textwidth]{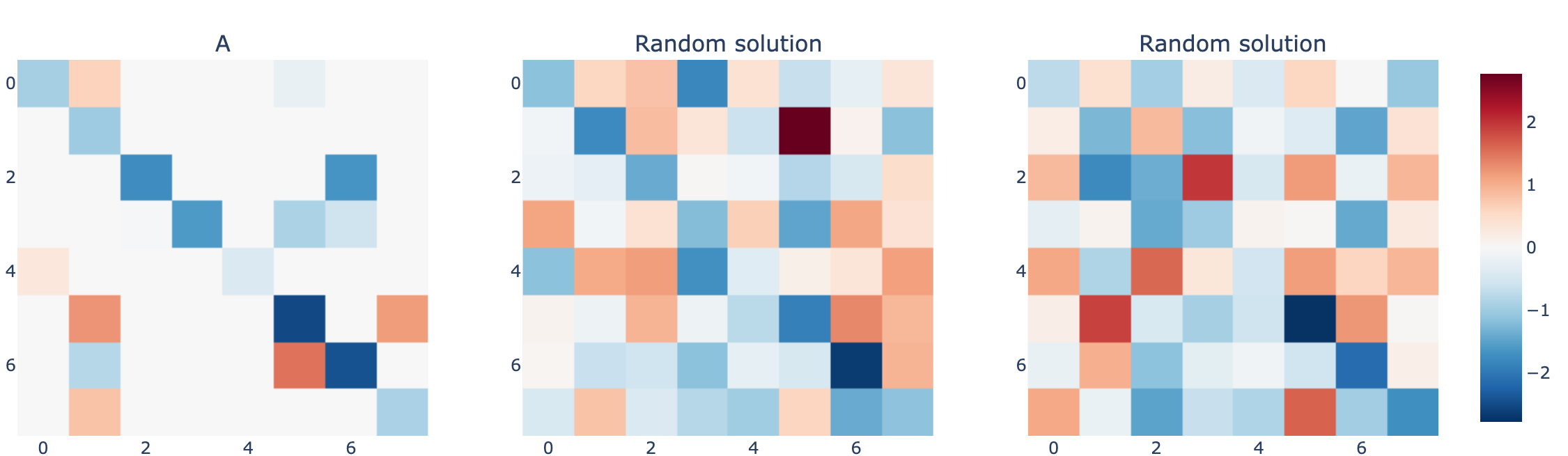}
    \caption{Heatmaps for three examples of weighted connectivity matrices of dimension $8\times 8$. The leftmost heatmap is for a sparse matrix $A\in \mathbb{R}^{8\times 8}$ generated randomly using the algorithm in Section \ref{sec:Val}. The other two matrices were selected randomly from the solution space $\mathcal{S}_\Gamma$ for the matrix $A$. This means the solution to the system in (\ref{eq:lin_sys}) has the same covariance matrix for all three matrices. 
    }
    \label{fig:rand_sols}
\end{figure}

It is known~\cite{socha2007linearization, appliedSDEs} that $A$ and $\Gamma$ satisfy the following Lyapunov equation:
\begin{equation}\label{eq:lyap_eq}
    \Gamma A^T + A \Gamma = -I.
\end{equation}

Given $A$, the stable covariance matrix solution is given by
$$
    \text{vec}(\Gamma) = -(I\otimes A + A\otimes I)^{-1}\text{vec}(I),
$$
which is readily computed. However, because the covariance matrix is symmetric, the inverse problem, i.e., solving for $A$ given $\Gamma$, does not have a unique solution~\cite{fernando1981solution}. The differential covariance approach~\cite{Lin_Das_Krishnan_Bazhenov_Sejnowski_2017} uses the fact that the expected value of $\dot x x^T$ is equal to $A\Gamma$ to obtain a unique inverse solution $A$. However, when measurement or intrinsic noise are large, computing time-series derivatives can be problematic or unfeasible.
Our goal is to propose a new algorithmic method that, under suitable assumptions, can single out a unique solution to the inverse covariance Lyapunov equation~\eqref{eq:lyap_eq}. Following \cite{fernando1981solution}, given $\Gamma$, define its orthonormal spectral decomposition:
\begin{equation}\label{eq:gamaDecomp}
    \Gamma = U  C  U^T
\end{equation}
where $U$ is an orthonormal matrix and $C$ is a diagonal matrix of entries $c_i$, the eigenvalues of $\Gamma$. Let $\bar A = U^T A U$. Then we can rewrite (\ref{eq:lyap_eq}) as
$$U^T(UCU^T \!A^T + AUCU^T)U=C\bar A^T + \bar AC=-I$$
which has solutions $\bar A$ easily shown to satisfy 
\begin{equation}\label{eq:A_cases}
    \bar A_{i,j} = \begin{cases}
            \frac{-1}{2c_i}, & \text{if } i = j \\
            -\frac{c_i}{c_j} \bar A_{j,i}, & \text{if } i \neq j
        \end{cases}.
\end{equation}

The set of matrices $\bar A$ satisfying~\eqref{eq:A_cases} is an affine subspace of dimension $m=\frac{n(n-1)}{2}$. Since $A = U \bar A  U^T$, this means that the set of possible solutions $A$ to the inverse covariance Lyapunov equations is also an $m$-dimensional affine space $\mathcal{S}_\Gamma\subseteq \mathbb{R}^{n\times n}$. Figure \ref{fig:rand_sols} shows a Hurwitz matrix $A$, with correlation matrix $\Gamma$, together with other two matrices taken randomly from the solution space $\mathcal{S}_\Gamma$, hence associated to the same covariance matrix. 

The following proposition characterizes the matrices in $\mathcal{S}_\Gamma$ with a linear equality constraint on the vectorized matrices. It is needed to formulate an $L_1$-regularization problem over $\mathcal{S}_\Gamma$ as a LP problem.
\begin{proposition}\label{prop:l1 to lp}
    Given a covariance matrix $\Gamma=UCU^T$, define a matrix $\mathbf{M}$ of dimensions $\frac{n(n+1)}{2}\times n^2$ and a vector $\mathbf b$ of dimension $\frac{n(n+1)}{2}$ with entries, respectively, given by
$$
    \mathbf M_{\iota_1(i,j),\iota_2(\ell, k)} = \begin{cases}
        c_j U_{\ell, i}U_{k,j} + c_i U_{\ell, j}U_{k,i} &\text{if } i\neq j \\
        U_{\ell, i}U_{k,j} &\text{if } i=j
    \end{cases}
$$
and
$$
    \mathbf b_{\iota_1(i,j)} = \begin{cases}
        0 &\text{if } i\neq j \\
        \frac{-1}{2c_i} &\text{if } i=j
    \end{cases}
$$
where $\iota_1:\{(i,j) \mid i,j\in [n], i \leq j\}\to[\frac{n(n+1)}{2}]$ and $\iota_2:[n]^2\to[n^2]$ are one-to-one functions that give the index of the vectorizations of the upper triangle of $\bar A$ (including the diagonal) and of the full matrix respectively. Then for $v\in\mathbb{R}^{n^2}$, we have that $\text{vec}^{-1}(v)\in \mathcal{S}_\Gamma$ if and only if $\mathbf M v = \mathbf b$. 
\end{proposition}
\begin{proof}
    Let $v\in \mathbb{R}^{n^2}$ and take $A=\text{vec}^{-1}(v)$. Suppose $A\in \mathcal{S}_\Gamma$. From~\eqref{eq:A_cases}, if $i\neq j$ then
    $\sum_{\ell = 1}^n\sum_{k = 1}^n U_{\ell, i}A_{\ell, k}U_{k,j} = (U^T A U)_{i,j} =\bar A_{i,j} = -\frac{c_i}{c_j}\bar A_{j,i} = -\frac{c_i}{c_j}\sum_{\ell = 1}^n\sum_{k = 1}^n U_{\ell, j}A_{\ell, k}U_{k,i}$,
    from which it follows that
    \begin{equation}\label{eq: isom proof 1}
        \sum_{\ell = 1}^n\sum_{k = 1}^n (c_j U_{\ell, i}U_{k,j} + c_i U_{\ell, j}U_{k,i})A_{\ell, k} = 0.
    \end{equation}
    Similarly, when $i=j$,
    \begin{equation}\label{eq: isom proof 2}
        \bar A_{i,j} = \sum_{\ell = 1}^n\sum_{k = 1}^n U_{\ell, i}U_{k,j}A_{\ell, k} = \frac{-1}{2c_i}.
    \end{equation}
    Equations~\eqref{eq: isom proof 1},\eqref{eq: isom proof 2} are precisely the $\frac{n(n+1)}{2}$ identities that are encoded in $\mathbf M$ and $\mathbf b$, which implies that $\mathbf M v = \mathbf b$.

    Following the same argument, for any $v\in\mathbb{R}^{n^2}$, if $\mathbf M v = \mathbf b$, then $\text{vec}^{-1}(v)\in\mathcal{S}_\Gamma$.
\end{proof}

\section{$L_1$-regularization on $\mathcal{S}_{\Gamma}$ and sparse network reconstruction}
\label{sec:L1Reg}

Assume that the interaction network we are reconstructing is sparse. Then, a reasonable way to recover a particular solution from $\mathcal{S}_\Gamma$ is to look for matrices with minimum $L_1$ norm; this is a common approach for promoting sparsity of solutions~\cite{candes2005l1}. We pose the problem of minimizing the $L_1$ norm of the vectorized matrices as follows:
\begin{align}
        \text{Minimize   } & \lVert v\rVert_1 \label{eq:l1_prob}\\
        \text{Subject to      } & \mathbf Mv=\mathbf b \,. \nonumber
\end{align}
In turn,~\eqref{eq:l1_prob} can be transformed into a Linear Programming (LP) problem by taking a new variable $u\in \mathbb{R}^{n^2}$ and a vector $\mathfrak{s}=(\mathbf{1}_{n^2}, \mathbf{0}_{n^2})$, and defining the LP problem:
\begin{align}
    &\text{Minimize   } & {\mathfrak{{s}}}^T&\begin{bmatrix} u \\ v \end{bmatrix} \label{eq:l1_lp_prob}\\
    &\text{Subject to} & \begin{bmatrix}
        \mathbf{0} & \mathbf M
    \end{bmatrix}&\begin{bmatrix} u \\ v \end{bmatrix} = \mathbf b \nonumber\\
    &\text{and} & \begin{bmatrix}
        -I & I \\
        -I & -I
    \end{bmatrix}&\begin{bmatrix} u \\ v \end{bmatrix} \leq \mathbf{0}_{2n^2}. \nonumber
\end{align}
The LP problem \eqref{eq:l1_lp_prob} minimizes the objective $\sum_{(i,j)} u_{(i,j)}$ subject to the restriction that $\text{vec}^{-1}(v)\in\mathcal{S}_\Gamma$ and $\lvert v_{(i,j)} \rvert\leq u_{(i,j)}$ for all $i,j\in[n]$. 
A solution to the LP problem (\ref{eq:l1_lp_prob}) immediately gives a solution to the problem (\ref{eq:l1_prob}) and \textit{vice-versa} -- in this sense the two problems are equivalent \cite{candes2005l1}.

Solving~\eqref{eq:l1_lp_prob} will return a sparse solution from $\mathcal{S}_\Gamma$. However, because there might be (and usually are) many sparse matrices in $\mathcal{S}_\Gamma$, the solution to the $L_1$-regularization problem might still be very different from the true network adjacency matrix, including under the sparsity assumption and mostly depending on the initial problem data. In the following section we present a way to constrain the optimization problem by including approximate prior knowledge on edge existence to bias the reconstruction toward solutions that are consistent with the priors. 

\section{Including priors on edge existence}
\label{sec:AddPriors}

\begin{figure}
    \centering
    \includegraphics[width = 0.49\textwidth]{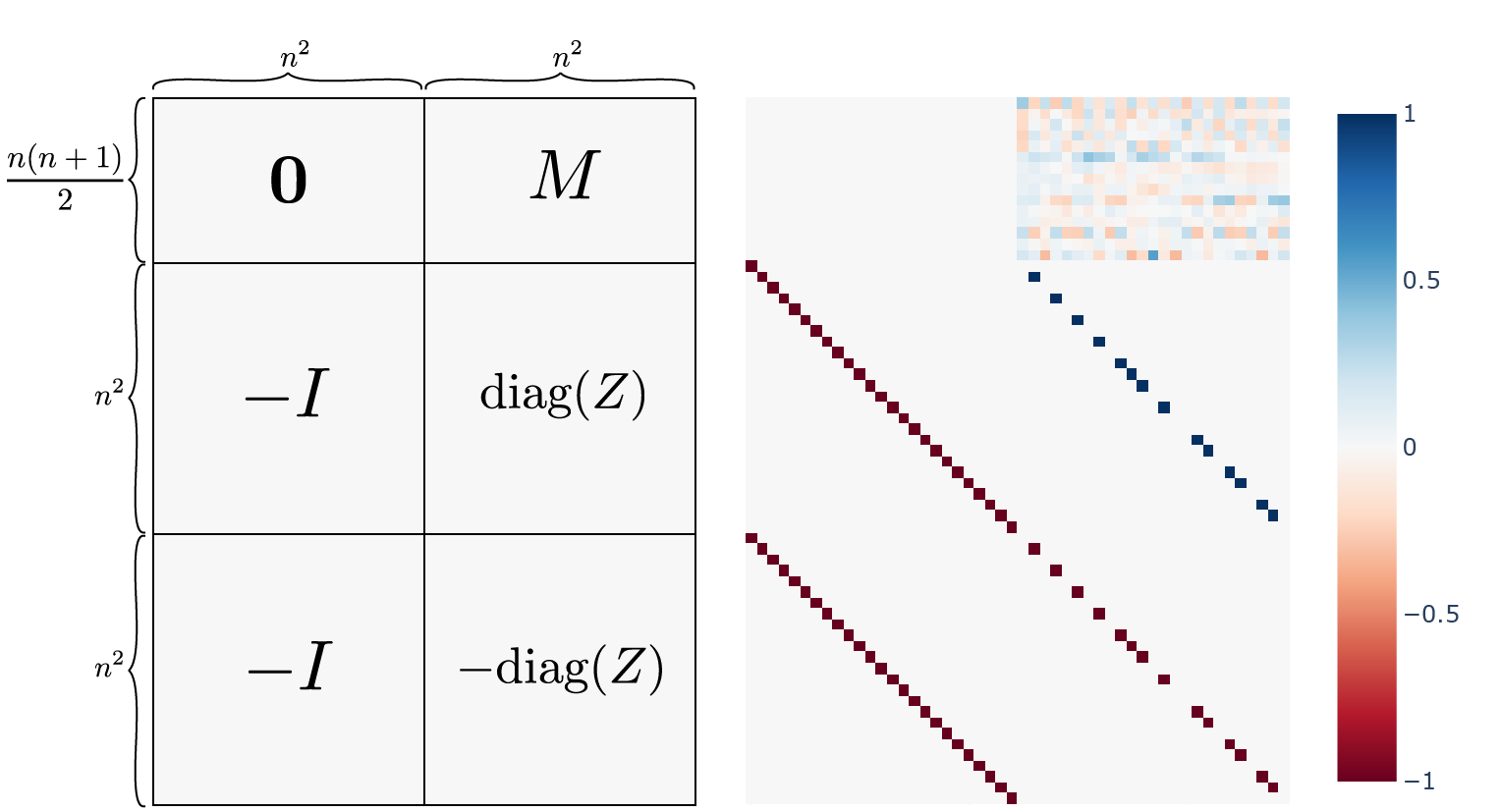}
    \caption{Diagram illustrating the way the optimization problem (\ref{eq:l1_prob}) can be encoded as a linear programming problem. The top $\frac{n(n+1)}{2}\times 2n^2$ sub-matrix is used as an equality constraint, which restricts the search to matrices on $\mathcal{S}_\Gamma$. The $2n^2\times 2n^2$ sub-matrix bellow is used as an inequality constraint which  ensures that for each $\ell, k$ we have $\lvert v_{(\ell, k)} \rvert \leq u_{(\ell, k)}$. An example of this encoding for the case $n=5$ is show on the right.}
    \label{fig:SCSreduction}
\end{figure}

The idea for including priors comes from the observation that if we knew that exactly $p$ directed edges (and which ones) existed in the network, then $A$ would lie in $\mathcal{R}_\Gamma\subseteq \mathbf{R}^{n\times n}$ of dimension $p$, i.e., the subspace parameterized by the $p$ nonzero elements of $A$. 
By Proposition~\ref{prop:l1 to lp}, $A\in \mathcal{S}_\Gamma$. Thus, $A\in \mathcal{S}_\Gamma\cap \mathcal{R}_\Gamma$, that is, $A$ must lie in the intersection of two affine subspaces of dimension $m=\dim(\mathcal{S}_\Gamma)=\frac{n(n-1)}{2}$ and $p=\dim(\mathcal{R}_\Gamma)$, respectively.

It is well known~(a consequence of Theorem 2.43 in\cite{axler2023linear}) that the transversal intersection of two linear subspaces of $\mathbb{R}^{n\times n}$ of dimension $p$ and $m$, with $m+p\leq n^2$ is either empty or a point.                            
This means that if we have \textit{exact} knowledge of which edges exist in the network and the network has a sufficiently small number of connections (i.e., $m+p\leq n^2$), then solving~\eqref{eq:l1_lp_prob} restricted to $\mathcal{R}_\Gamma$ would return $A$ {\it exactly}. However, given only imperfect knowledge of edge existence and if $m+p\leq n^2$, then our noisy approximation $\hat{\mathcal{R}}_\Gamma$ of $\mathcal{R}_\Gamma$ and $\mathcal{S}_\Gamma$ will almost surely not intersect at all. In spite of this, looking for the matrix $A\in \mathcal{S}_\Gamma$ with minimal distance to $\hat{\mathcal{R}}_\Gamma$ should bias the solution~\eqref{eq:l1_lp_prob} toward the real $A$. This is the approach we use here.
 

Based on this idea we define a vector with entries representing the relative importance 
assigned to each connection during the optimization, $Z\in[0,1]^{n^2}$. The objective of the minimization then becomes $\sum_{(i,j)}Z_{(i,j)}\lvert v_{(i,j)} \rvert$. In the simplest case, we can introduce information about the presence of edges in the following way. Let $\mathcal{I}=\{(i,j) \in [n]^2 \mid (i,j) \text{ is a known edge}\}$. Then we can define $Z$ as 
$$
    Z_{(i, j)} = \begin{cases}
        1 &\text{if } (i,j)\notin\mathcal{I} \\
        0 &\text{otherwise.}
    \end{cases}
$$
Written as an LP the problem becomes
\begin{align}
    &\text{Minimize   } & {\mathfrak{s}}^T&\begin{bmatrix} u \\ v \end{bmatrix} \label{eq:l1_lp_prob_wprior}\\
    &\text{Subject to} & \begin{bmatrix}
        \mathbf{0} & \mathbf M
    \end{bmatrix}&\begin{bmatrix} u \\ v \end{bmatrix} = \mathbf b \nonumber\\
    &\text{and} & \begin{bmatrix}
        -I & \text{diag}(Z) \\
        -I & -\text{diag}(Z)
    \end{bmatrix}&\begin{bmatrix} u \\ v \end{bmatrix} \leq \mathbf{0}_{2n^2}. \nonumber
\end{align}
When there are no edge priors, $Z=\mathbf{1}_{n^2}$, and the problem in (\ref{eq:l1_lp_prob_wprior}) reduces to (\ref{eq:l1_lp_prob}). An illustrative example of the constraint matrix in (\ref{eq:l1_lp_prob_wprior}) is shown in Figure \ref{fig:SCSreduction}. Note in Figure \ref{fig:SCSreduction} that this matrix is sparse in most places except the top right corner; thus, optimization methods that take advantage of sparse structure can be employed. Problem (\ref{eq:l1_lp_prob_wprior}) leads to an optimization that prioritizes finding a solution where the edges {\it not in} $\mathcal{I}$ have weight close to zero in absolute value. The values of $Z$ could also be adjusted to correspond with the level of certainty about prior edge existence. For simplicity, in the following we only consider $Z\in\{0,1\}^{n^2}$. 

\section{Edge priors from Transfer Entropy}
\label{sec:TE}


There are various ways to obtain prior knowledge on edge existence, e.g., from expert knowledge or from previous analysis. One generally applicable tool for inferring the existence of directed interactions is Transfer Entropy (TE)\cite{Schreiber_2000}. Given two discrete-time stationary stochastic processes $X$ and $Y$, the TE from $X$ to $Y$ is a measure of the predictive power of $X$ about $Y$ at a \textit{discrete} time $\tt{t}$. In its simplest form TE is defined as \cite{Schreiber_2000}
\begin{align*}
     &T_{X\to Y}({\tt t}) = I(X({\tt t}-1) ; Y({\tt t}) \mid Y({\tt t}-1)) \\ 
               &= H(Y({\tt t})\mid Y({\tt t}-1)) - H(Y({\tt t}) \mid X({\tt t}-1), Y({\tt t}-1)).
\end{align*}
So, $T_{X\to Y}(\tt t)$ is the mutual information $I$ of $X({\tt t}-1)$ and $Y({\tt t})$ conditioned on $Y({\tt t}-1)$, where
$$
    {\scriptstyle I(X;Y\mid Z) = \int_{\mathcal{Z}}\int_{\mathcal{Y}}\int_{\mathcal{X}}p_{X,Y,Z}(x,y,z)\ln\left(\frac{p_Z(z)p_{X,Y,Z}(x,y,z)}{p_{X,Z}(x,z)p_{Y,Z}(y,z)}\right)dx\;dy\;dz},
$$
$p$ are probability density functions of the joint distributions of variables represented by the subscripts, and $\mathcal{X},\mathcal{Y}, \mathcal{Z}$ are the supports of $X,Y,Z$ respectively. $T_{X \to Y}({\tt t})$ can equivalently be written in terms of $H(X\mid Y)=-\int_{\mathcal{X}}\int_{\mathcal{Y}}p_{X,Y}(x,y)\ln(\frac{p_{X,Y}(x,y)}{p_Y(y)})$. If the processes are assumed to be stationary then TE is the same for all values of $\tt t$ and we can simply write $T_{X\to Y}$.

This measure has been shown to be a good indicator of directed interactions between pairs of noisy dynamical variables~\cite{Bossomaier_Barnett_Harré_Lizier_2016, Rahimzamani_Kannan_2016, Novelli_Wollstadt_Mediano_Wibral_Lizier_2019}, and can be seen as a nonlinear generalization of Granger causality \cite{granger_te_2009}. It can also be extended to take into account factors such as redundancy and synergy in a partial information decomposition approach~\cite{Williams_Beer_2011}. 

A simple way to reduce the effects of redundancy when inferring connectivity is to use {\it conditional} TE: $T_{X\to Y\mid Z}$. The definition of $T_{X\to Y\mid Z}$ is the same as for $T_{X\to Y}$ but conditioned on the past of $Z$. Connectivity can then be inferred from a set of $n$ stochastic processes $X_i$ by calculating $T_{X_i\to X_j\mid X_{(i,j)}}$ for each pair $1\leq i,j\leq n$, where $X_{(i,j)}$ is the joint distribution of all the variables except $X_i$ and $X_j$. Then we say there is a significant interaction from $X_i$ to $X_j$ if $T_{X_i\to X_j\mid X_{(i,j)}}$ is higher than what would be expected by pure chance (using suitable data shuffling and significance testing).
However, this approach is usually not feasible for large networks due to the amount of required data.

To bypass this problem in cases where the number of edges in the network is small (i.e., the connectivity matrix is sparse), Novelli et al. \cite{Novelli_Wollstadt_Mediano_Wibral_Lizier_2019} proposed a greedy algorithm that finds a set of sources for each target node by {\it i)} creating a candidate set of source nodes with significant TEs when conditioned only on previously selected nodes and {\it ii)} removing all nodes for which TE became non-significant due to redundancy with newly added nodes.
This algorithm (and others) is implemented in the IDTxl Python package \cite{Wollstadt_Lizier_Vicente_Finn_Martínez-Zarzuela_Mediano_Novelli_Wibral_2019}.

In order to keep inference fast, we adapted the greedy algorithm in \cite{Novelli_Wollstadt_Mediano_Wibral_Lizier_2019} to a simplified version where only Step i) of the algorithm is used.
This did not seem to drastically impact the performance of the method. However, we observed an increase in false positives, so we introduced a strict upper bound on the number of inferred edges.

\section{Numerical validation}
\label{sec:Val}

\begin{figure}
    \centering
    \includegraphics[width=0.47\textwidth]{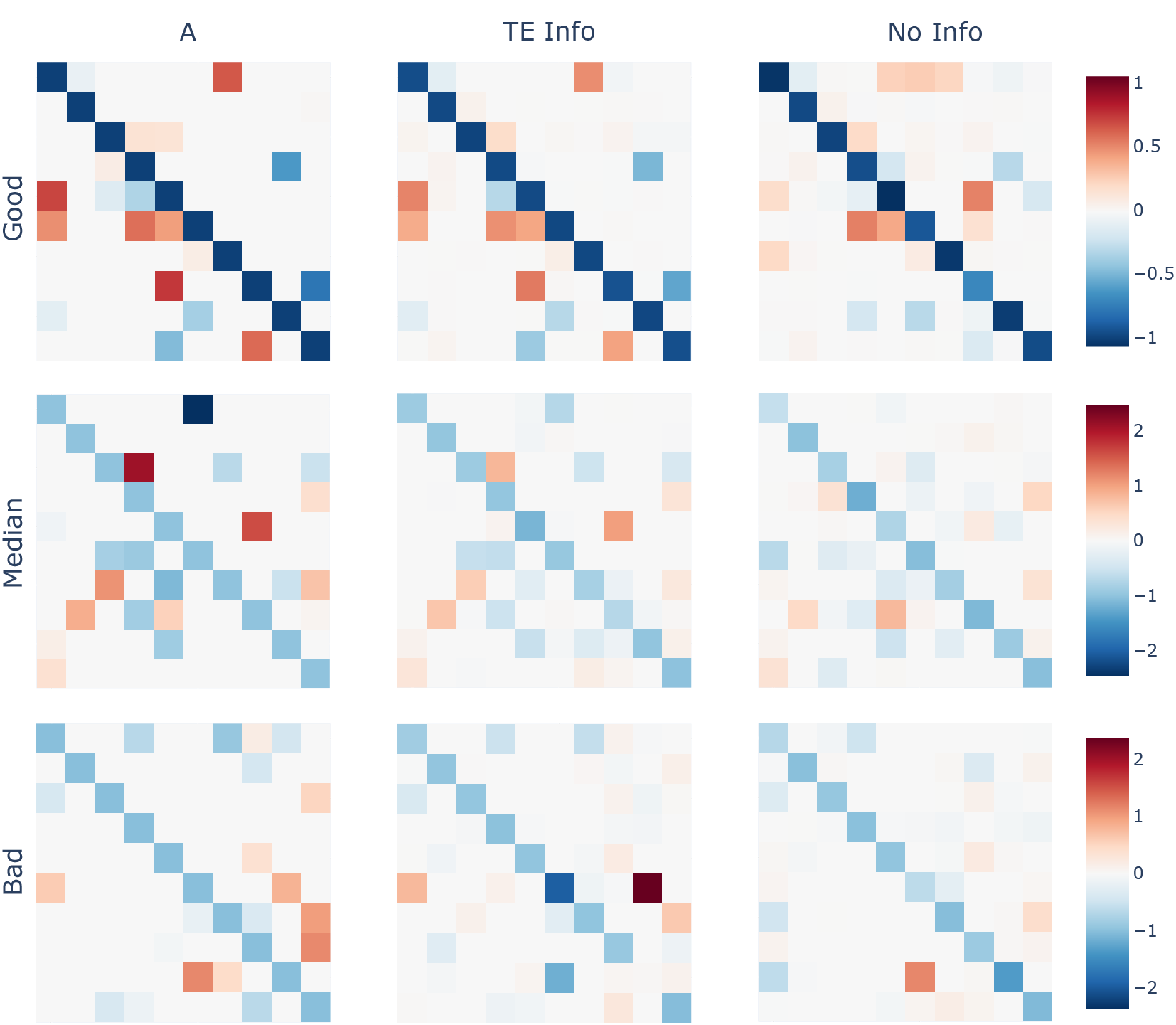}
    \caption{Three examples of reconstructed connected networks using our method with edge priors inferred using TE (middle), and no info (right),  applied to time-series generated from the weakly-nonlinear system in (\ref{eq:sat_sys}). The ground truth connectivity matrix $A$ is show on the left. The top row shows examples of good reconstruction (high alignment between reconstructed matrix and the ground truth), the middle row shows reconstruction with median alignment, and the bottom row shows a bad reconstruction (low alignment). 
    The alignments for the reconstructions with TE Info were $0.99$, $0.90$ and $0.31$ for the good, median, and bad respectively. For the No Info the alignments were $0.43$, $0.42$, and $0.53$ for the good, median, and bad respectively. Examples were selected from $600$ simulated dynamics of matrices with $20$ edges and different $\varepsilon$ values by sorting according to alignment of TE Info reconstruction and selecting the best, median, and worst reconstructions.}
    \label{fig:goodbad_examples}
\end{figure}

\begin{figure}
    \centering
    \includegraphics[width=0.45\textwidth]{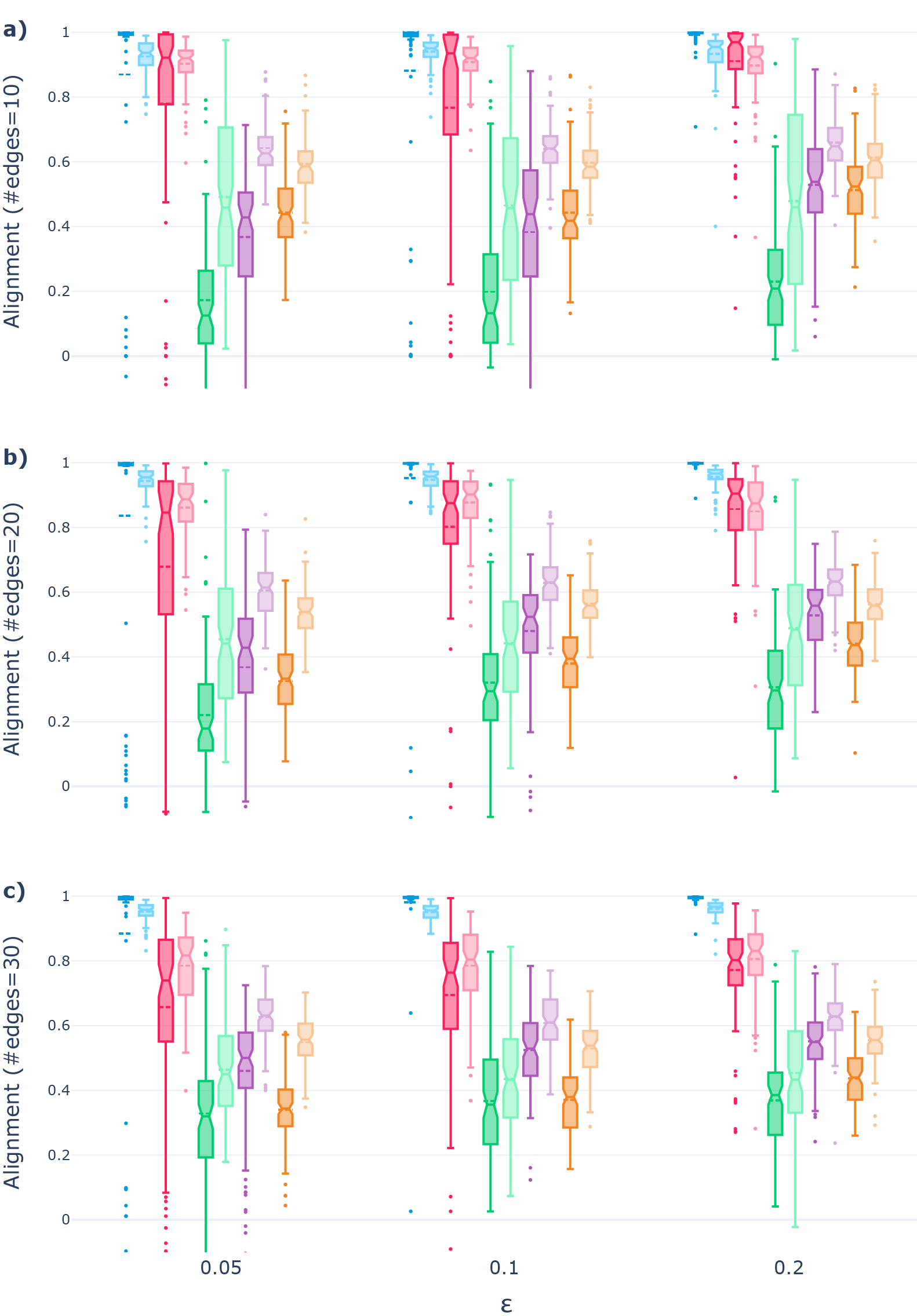}
    \caption{Plots showing the performance of our method for $100$ samples of random Hurwitz matrices of size $10\times 10$ generated with different values of $\varepsilon$, and with \textbf{a)} $10$, \textbf{b)} $20$, \textbf{c)} $30$ edges. The $L_1$ optimization methods compared are: Full Info (blue), TE Info (red), No Info (green) on edge existence. Performance of precision and correlation matrices were added as reference in purple and orange respectively. Darker colors correspond to reconstructions for data simulated using the linear system in (\ref{eq:lin_sys}), and lighter colors for the weakly-nonlinear system in (\ref{eq:sat_sys}). Notches in boxes represent median values and dotted lines mean values. Boxes encompass $50\%$ of the data points. Isolated points are outliers. }
    \label{fig:epsVSalign}
\end{figure}

We tested our method on randomly generated matrices of the form
$$
    A = -I + \frac{1-\varepsilon}{b_{max}}B \, ,
$$
where $B$ is is an Erdős–Rényi graph with normally distributed edge weights, $\varepsilon\in(0,1]$ is a parameter that controls the spectrum of $A$ by weighting $B$, and $b_{max}=\max\{\Re(\lambda)\mid \lambda\in\sigma(B)\}$. We get $a_{max}=\max\{\Re(\lambda)\mid \lambda\in\sigma(A)\}=-1+\varepsilon$.

\subsection{Implementation}

For numerical validation, our algorithm was implemented in standard Python libraries. The code can be found at \textit{https://github.com/ianxul/SDE-net-reconstruction}. For estimating conditional TE interactions, we used the GPU-enabled estimator in IDTxl package \cite{Wollstadt_Lizier_Vicente_Finn_Martínez-Zarzuela_Mediano_Novelli_Wibral_2019}, and the adapted algorithm described in Section~\ref{sec:TE}.
We further simplified the permutation-based significance test suggested in~\cite{Novelli_Wollstadt_Mediano_Wibral_Lizier_2019} by using a less computationally intensive threshold heuristic similar to~\cite{Rahimzamani_Kannan_2016}. $L_1$-optimization was performed using CVXPY package~\cite{diamond2016cvxpy} with the splitting conic solver (SCS).

Simulations of SDEs were run with Julia's DifferentialEquations solver. We used the LambaEM solver for SDEs with a step size of $0.1$ for $10,000$ time steps (giving $100,000$ data points for each simulation).
We applied our method to the data generated from the linear SDE system~(\ref{eq:lin_sys}) and from the weakly nonlinear (saturated network interactions) SDE system
\begin{equation}
    \label{eq:sat_sys}
    \frac{dx_i}{dt} = A_{ i,i } x_i + \sum_{ i \neq j}( \tanh(A_{ i,j } x_j) ) + w(t).
\end{equation}

To compare different methods we computed the alignment between the vectorized matrices, excluding the diagonal, as it is not related to network reconstruction:
\begin{equation}
    \label{eq:alignment}
    \text{align}(A,B)=\frac{\sum_{\substack{i\neq j}} A_{i,j}B_{i,j}}{\left(\sum_{\substack{i\neq j}} A_{i,j}^{2}\right)^\frac{1}{2}\left(\sum_{\substack{i\neq j}} B_{i,j}^2\right)^\frac{1}{2}}.
\end{equation}

We compared three cases: {\it i)} ``Full Info'', in which full information on edge existence is given and which should lead to perfect reconstruction; {\it ii)} ``TE Info'', in which edge existence is inferred using TE methods; {\it iii)} ``No Info'', in which we simply perform $L_1$-regularization over $\mathcal{S}_\Gamma$ as in (\ref{eq:l1_lp_prob}). Results are exemplified in Figure~\ref{fig:goodbad_examples} and summarized in Figure~\ref{fig:epsVSalign}.

\subsection{Results and discussion}

Our results show that adding edge existence priors obtained through TE \cite{Novelli_Wollstadt_Mediano_Wibral_Lizier_2019} to our reconstruction method improves the alignment of reconstructions considerably when compared against doing $L_1$-regularization with no prior information, and against the precision and correlation matrices, as can be seen in Figures~\ref{fig:goodbad_examples} and~\ref{fig:epsVSalign}. Interestingly, $L_1$-regularization with no edge priors performed worse than the symmetric methods for both the linear and weakly nonlinear cases.
A possible reason is that for a given covariance matrix the solution space $\mathcal S_\Gamma$ might contain many matrices with small $L_1$ norm and, due to noise, one of these (potentially very far from the true solution) can turn into a global minimum.
We observed, for instance, cases where edges were flipped (see the top right matrix in Figure \ref{fig:goodbad_examples}). By biasing the optimization towards matrices consistent with prior knowledge we increase the likelihood that the correct sparse solution is selected.

All methods performed increasingly worse as the ground truth solution $A$ became less sparse or as $\varepsilon$ became smaller (i.e., as the spectrum of $A$ got closer to the imaginary axis).
This could simply be a result of not having an optimal algorithm for inferring edges using TE, since the performance of the method with full prior knowledge of existing edges (``Full Info") is mostly unaffected by changes in $\varepsilon$ and edge number.
Applying the full IDTxl method\cite{Novelli_Wollstadt_Mediano_Wibral_Lizier_2019} could improve results further due to more rigorous significance tests, a redundant edge cleanup pass and no hard limit on inferred edges. The mean and median performance of the ``TE Info" method was always strictly better than all other methods we tested (apart, of course, from ``Full Info") by a significant margin (bootstrap test $p<10^{-5}$).

Another interesting observation is that introducing a weak nonlinearity in the model usually led to only minor changes in the median performance but often drastically reduced the performance variance, and in some cases increased the median performance as well. In other words, our method performed more consistently (and sometimes better) in the presence of weakly nonlinear network interactions as compared to a purely linear network model. This effect cannot be explained by the use of nonlinear edge inference methods like TE, since fully linear methods like ``No Info" and those based on precision and correlation matrices also benefited from the presence of weakly nonlinear interactions. We will investigate this interesting phenomenon in future research.

With respect to computational complexity, our model is fast and high performing since it has the computational complexity of LP optimization (not including the prior edge inference with TE, which is highly parallelizable \cite{Wollstadt_Lizier_Vicente_Finn_Martínez-Zarzuela_Mediano_Novelli_Wibral_2019}). This has been shown to be theoretically less than or equal to the current complexity of matrix multiplication \cite{Cohen_Lee_Song_2021}, which is $\approx O(\#^{2.3})$, where $\#$ is the number of variables in the optimization. In our case, the number of variables is $\# = 2n^2$, and so the (theoretical) computational complexity of the method is $\approx O(n^{4.6})$. For reference, solving the LP optimization for a $100$-dimensional system takes about $8$ minutes on a single core CPU node with 2.6 GHz Intel Skylake processor and $16$ GB of memory.
 
\section{Conclusions and perspectives}
\label{sec: conclusions}

We introduced a fast and high performing method to reconstruct directed and weighted networks and analyzed its performance using both linear and weakly nonlinear network models. Future works will aim at improving the heuristics used for the TE inference step. Theoretically it will also be important to understand the role of nonlinearities on network inference, which might be key for application to (usually nonlinear) neuronal and other kind of biological networks. Measures of the intrinsic timescale of each node, estimated through auto-correlation analysis~\cite{rossi2021invariant}, can further be used to further constraint the $L_1$-optimization step. This will also require consideration of the role of network self-loops, i.e., the diagonal elements of $A$, more explicitly in the method. All these ideas will soon be implemented and tested on experimental neuroscience (EEG) datasets.

\bibliographystyle{IEEEtran}
\bibliography{Ref}


\end{document}